\documentclass[12pt, reqno]{amsart}
\usepackage{tikz-cd}
\usepackage{amssymb}
\usepackage{amsmath}
\allowdisplaybreaks[1]
\usepackage{graphicx,color}
\usepackage{wrapfig,framed,bbm}
\usepackage{mathtools}
\usepackage[height=22.0cm, width=16.3cm, hmarginratio={1:1}]{geometry}
\usepackage{hyperref}

\usetikzlibrary{decorations.markings}
\tikzset{double line with arrow/.style args={#1,#2}{decorate,decoration={markings,%
mark=at position 0 with {\coordinate (ta-base-1) at (0,1pt);
\coordinate (ta-base-2) at (0,-1pt);},
mark=at position 1 with {\draw[#1] (ta-base-1) -- (0,1pt);
\draw[#2] (ta-base-2) -- (0,-1pt);
}}}}

\newcommand{\p}{\partial}
\newcommand{\nn}{\nonumber}
\newcommand{\e}{\epsilon}

\newcommand{\bt}{{\bf t}}
\newcommand{\F}{\mathcal{F}}
\newcommand{\CC}{\mathbb{C}}

\newcommand{\beq}{\begin{equation}}
\newcommand{\eeq}{\end{equation}}

\numberwithin{equation}{section}
\theoremstyle{plain}
\newtheorem{theorem}{Theorem}

\newtheorem{lemma}{Lemma}

\theoremstyle{definition}

\begin{document}
\title[Loop Equations]{GUE via Frobenius Manifolds. II. Loop Equations}
\author{Di Yang}
\address{School of Mathematical Sciences, University of Science and Technology of China, Hefei 230026, P.R.~China}
\email{diyang@ustc.edu.cn}
\date{}
\begin{abstract}
A theorem of Dubrovin establishes the 
relationship between the GUE partition function and the partition function of 
Gromov-Witten invariants of the complex projective line.
Based on this theorem we derive loop equations for the Gaussian Unitary Ensemble (GUE) partition function. 
We show that the GUE partition function is equal to part of the topological partition function of 
the non-linear Schr\"odinger (NLS) Frobenius manifold.  
\end{abstract}
\maketitle

\setcounter{tocdepth}{1}
\tableofcontents

\section{Introduction}
In this article, which is the second one in the series, 
we derive the loop equations for the Gaussian Unitary Ensemble (GUE) partition function.
By the GUE partition function, we mean the following integral~\cite{AvM, Du09, DY17, Y2} (see also~\cite{FGZ, KKN, mehta, Mo94}):
\beq\label{globaldefZ}
Z(x,{\bf s};\e) = \frac{(2\pi)^{-n}\e^{-\frac1{12}}}{{\rm Vol}(n)} \int_{{\mathcal H}(n)} e^{-{\frac1\e}  {\rm tr} \, Q(M; {\bf s})} dM, \quad x=n\e,
\eeq
where ${\bf s}:=(s_1,s_2,\cdots)$, 
${\rm Vol}(n) = \pi^{\frac{n(n-1)}2}/\prod_{j=1}^{n-1} j!$,  
$Q(m; {\bf s}) = \frac12 m^2 -\sum_{j\geq 1} s_{j} m^{j}$,
and $dM = \prod_{1\leq i\leq n} d M_{ii} \prod_{1\leq i<j\leq n} d{\rm Re} M_{ij} d{\rm Im}M_{ij}$.
Here $x$ denotes the {\it t'Hooft coupling constant}~\cite{thooft1,thooft2}.

Recall that the logarithm of the GUE partition function $\log Z(x,{\bf s};\e)=:\F(x,{\bf s};\e)$, called the 
{\it GUE free energy}, has the following genus expansion:
\beq
\F(x,{\bf s};\e) =: \sum_{g\ge0} \e^{2g-2} \F_g(x, {\bf s}),
\eeq
where $\F_g(x,{\bf s})$ is called the {\it genus $g$ GUE free energy}. More explicitly, we have~\cite{BIZ, HZ} 
(cf.~\cite{Barnes, Du09, DY17, EM, FL, WW, Y2})
\begin{align}
\F_g(x, {\bf s}) = & \frac{x^2}{2\e^2}\Bigl(\log x-\frac32\Bigr) \delta_{g,0} + 
 \Bigl(\zeta'(-1) - \frac{\log x}{12}\Bigr) \delta_{g,1} + \delta_{g\geq2} \frac{\e^{2g-2} B_{2g}}{4g(g-1)x^{2g-2}} \nn \\
& + \sum_{k\geq 1} \sum_{ j_1, \dots, j_k\geq 1\atop 2-2g - k + |{\bf j}|/2\ge1} a_g({\bf j}) 
 s_{j_1} \cdots s_{j_k} \e^{2g-2} x^{2-2g - k + \frac{|{\bf j}|}2}.
\end{align}
Here, for ${\bf j}=(j_1,\dots,j_k)$, $a_g({\bf j}) := \sum_{G} \frac{j_1 \cdots j_k}{|{\rm Aut} (G)|}$ with 
 the summation taken over all connected ribbons graphs of genus $g$ with $k$ vertices of valencies $j_1,\dots,j_k$.

In~\cite{Y2}, we gave a proof of a theorem of Dubrovin 
regarding the relationship between the GUE partition function and the partition function for the Gromov--Witten (GW)
invariants of~$\mathbb{P}^1$. The latter can be computed by solving  
the Dubrovin--Zhang (DZ) loop equation for the $\mathbb{P}^1$ Frobenius manifold. So   
one can use this loop equation for computing the GUE partition function. 
By the Legendre-type transformation~\cite{Du96}, 
the $\mathbb{P}^1$ Frobenius manifold has a correspondence to the so-called NLS Frobenius manifold~\cite{Du96}. 
Here ``NLS" stands for non-linear Schr\"odinger.   
To get a better geometric picture for GUE, we derive 
in this paper another loop equation which 
 turns out to be the same as the DZ loop equation for the NLS Frobenius manifold. 
Based on this we identify the GUE partition function 
with part of the topological partition function of the NLS Frobenius manifold, as originally observed by Dubrovin~\cite{Du09}.
 
Let us start with recalling the definition of a Frobenius manifold~\cite{Du93, Du96}, which is now also known as Dubrovin--Frobenius manifold.
The notion of Frobenius manifold was introduced by Dubrovin~\cite{Du93,Du96} to give a coordinate free description 
 for the WDVV (Witten--Dijkgraaf--Verlinde--Verlinde) equations that appeared in the 2D topological field theories. 
 Recall that a {\it Frobenius algebra} is a triple $(V, \, e, \, \langle\,,\,\rangle)$, 
where $V$ is a commutative and associative algebra over~$\CC$
with unity~$e$, and $\langle\,,\,\rangle: V\times V\rightarrow \mathbb{C}$ is a symmetric 
non-degenerate bilinear product satisfying 
$\langle x\cdot y , z\rangle = \langle x, y\cdot z\rangle$, $\forall\,x,y,z \in V$. 
A {\it Frobenius structure} of charge~$d$ on an $n$-dimensional complex manifold~$M$ is a family of 
Frobenius algebras $(T_p M, \, e_p, \, \langle\,,\,\rangle)_{p\in M}$, 
depending holomorphically on~$p$ and satisfying the following three axioms:
\begin{itemize}
\item[{\bf FM1}]  The metric $\langle\,,\,\rangle$ is flat; moreover, 
denote by~$\nabla$ the Levi--Civita connection of~$\langle\,,\,\rangle$, then it is required that 
\beq\label{flatunity511}
\nabla e = 0.
\eeq
\item[{\bf FM2}] Define a $3$-tensor field~$c$ by 
$c(X,Y,Z):=\langle X\cdot Y,Z\rangle$, for $X,Y,Z$ being holomorphic vector fields on~$M$. 
Then the $4$-tensor field $\nabla c$ is required to be 
symmetric. 
\item[{\bf FM3}]
There exists a holomorphic vector field~$E$ on~$M$ satisfying
\begin{align}
&  [E, X\cdot Y]-[E,X]\cdot Y - X\cdot [E,Y] = X\cdot Y, \label{E2}\\
&  E \langle X,Y \rangle - \langle [E,X],Y\rangle - \langle X, [E,Y]\rangle 
= (2-d) \langle X,Y\rangle. \label{E3}
\end{align}
\end{itemize}
The quintuple $(M,\, \cdot,\, e,\, \langle \,,\, \rangle,\, E)$ is  
called an {\it $n$-dimensional Frobenius manifold}, and the vector field~$E$ is called the {\it Euler vector field}.

A non-zero element $a$ of a Frobenius algebra is called a {\it nilpotent} if there exists $m\geq2$ such that $a^m=0$.
A Frobenius algebra is called {\it semisimple} if it contains no nilpotents. A point~$p$ on a Frobenius 
manifold~$M$ is called a {\it semisimple point} if the Frobenius algebra $T_p M$ is semisimple. A Frobenius manifold is 
called {\it semisimple} if its generic points are semisimple.

Starting with a Frobenius manifold~$M$, an integrable hierarchy of partial differential equations  
(PDEs) of hydrodynamic type, called the {\it principal hierarchy}, was defined~\cite{Du93, Du96}, and the 
 genus zero topological free energy $\F^{\rm top}_0$ for the principal hierarchy was constructed~\cite{Du93, Du96}.
It was shown by Dubrovin and Zhang that $\F^{\rm top}_0$ satisfies the genus zero Virasoro constraints~\cite{DZ99} (see also~\cite{EHX, Getzler, LT}). 
For the case when $M$ is semisimple, Dubrovin and Zhang~\cite{DZ-norm} (cf.~\cite{Du14}) 
 derived the loop equation for~$M$, and by solving this equation gave a way of 
reconstructing the higher genus topological free energies $\F^{\rm top}_g$, $g\ge1$. The exponential 
\beq
Z^{\rm top} = \exp \biggl( \,\sum_{g\ge0} \e^{2g-2} \F^{\rm top}_g \biggr)
\eeq
is called the {\it topological partition function of~$M$}.

In this paper we focus on the above-mentioned two important Frobenius manifolds:  
the NLS Frobenius manifold, denoted by 
$M_{\rm NLS}$, and the $\mathbb{P}^1$ Frobenius manifold, denoted by 
$M_{\rm \mathbb{P}^1}$. Both are two-dimensional and semisimple, and 
their potentials (see in Appendix~\ref{appa} for the definition of a potential of a Frobenius manifold) are given by 
\beq\label{FNLS27}
F_{M_{\rm NLS}} = \frac12 \varphi^2 \rho + \frac12 \rho^2 \log \rho - \frac34 \rho^2
\eeq
and 
\beq\label{FP11009}
F_{M_{\mathbb{P}^1}} = \frac12 v^2 u + e^u,
\eeq
respectively. It was found in~\cite{Du96} that they are related by the 
Legendre-type transformation~\cite{CDZ, Du96, DZ-norm}.
The principal hierarchy associated to $M_{\rm NLS}$ is equivalent to the dispersionless limit 
of the extended nonlinear Schr\"odinger hierarchy \cite{CDZ, CvdLPS, Du96, FY}), 
and the principal hierarchy associated to~$M_{\mathbb{P}^1}$ coincides with the dispersionless limit of the 
extended Toda hierarchy~\cite{CDZ, Du96, DZ-norm}.

The following theorem was shown in~\cite{DY17} (see also~\cite{AvM, Deift, GMMMO91, MMMM91, Mo94, Y2}).

\begin{theorem}[\cite{DY17}] \label{thmguetoda}
The GUE partition function $Z(x, {\bf s}; \e)$ is a particular tau-function in the sense of~\cite{CDZ, DY17, DZ04} for the 
Toda lattice hierarchy. 
\end{theorem}

Based on the DZ approach and Theorem~\ref{thmguetoda}, we gave~\cite{Y2} a proof of the following theorem, which 
gives the relationship between the GUE partition function and $M_{\rm \mathbb{P}^1}$.

\smallskip

\noindent {\bf Theorem A} (Dubrovin~\cite{Du09}) {\it 
For $g=0$, we have
\begin{align}
\mathcal{F}_0(x,{\bf s}) = \, &
\frac12\sum_{p,q\geq0} (p+1)! (q+1)! \Bigl(s_{p+1}-\frac12\delta_{p,1}\Bigr) \Bigl(s_{q+1}-\frac12\delta_{q,1}\Bigr) 
\Omega^{M_{\mathbb{P}^1},[0]}_{2,p;2,q}({\bf v}(x,{\bf s}))  \label{F0dub} \\
& + x\sum_{p\geq0} (p+1)! \Bigl(s_{p+1}-\frac12\delta_{p,1}\Bigr) \theta^{M_{\mathbb{P}^1}}_{2,p}({\bf v}(x,{\bf s}))  + \frac12 x^2 u(x,{\bf s}).  \nn
\end{align}
For $g\ge1$, the genus $g$ GUE free energy has the following representation:
\beq
\F_g(x,{\bf s}) = 
F_g^{M_{\mathbb{P}^1}}\Bigl(v(x,{\bf s}), u(x,{\bf s}), \dots, \frac{\p^{3g-2} v(x,{\bf s})}{\p x^{3g-2}}, \frac{\p^{3g-2} u(x,{\bf s})}{\p x^{3g-2}}\Bigr) + 
\delta_{g,1}\bigl(\zeta'(-1)-\tfrac1{24}\log(-1)\bigr). \label{jetguep1}
\eeq
Here
$\Omega^{M_{\mathbb{P}^1},[0]}_{\alpha,p;\beta,q}({\bf v})$ denote the two-point 
correlation functions for the principal hierarchy of~$M_{\mathbb{P}^1}$, 
the power series ${\bf v}(x,{\bf s})$ is defined in~\eqref{defvxs1214}--\eqref{defuxs1214}, and 
$F_g^{M_{\mathbb{P}^1}}(v, u, v_1,u_1, \dots, v_{3g-2}, u_{3g-2})$, $g\geq1$, 
denote the genus $g$ free energies in jets of the $\mathbb{P}^1$ Frobenius manifold.}

\smallskip

We recall that
for $g=1$, 
$$F_1^{M_{\mathbb{P}^1}}=\frac1{24} \log\bigl(v_1^2 - e^u u_1^2\bigr) - \frac1{24} u,$$
and for $g\geq2$, $F_g^{M_{\mathbb{P}^1}}$ has polynomial dependence in $v_2, u_2, \dots, v_{3g-2}, u_{3g-2}$
and rational dependence in $v_1,u_1$ \cite{DZ-norm} with 
coefficients being smooth functions of $v,u$.

The goal of this paper is to show that $Z(x,{\bf s};\e)$ is equal to part of the topological partition function of the NLS Frobenius manifold, 
which is given in Theorem~\ref{thmmain} below.

The rest of the paper is organized as follows. In Section~\ref{section2} we review the 
DZ loop equation for the NLS Frobenius manifold. In Section~\ref{section3} we 
prove Theorem~\ref{thmmain}. In Section~\ref{section4} we give concluding remarks. 
Review on the topological partition function of a semisimple Frobenius manifold is given in Appendix~\ref{appa}. 

\smallskip
\noindent {\bf Acknowledgements.} 
The work is partially supported by NSFC No.~12371254 and by the CAS Project for Young Scientists in Basic Research No. YSBR-032.

\section{The Dubrovin--Zhang loop equation for the NLS Frobenius manifold}\label{section2}
In this section, we give the explicit form of the DZ loop equation for the NLS Frobenius 
manifold $M_{\rm NLS}$. Note that we only need to specialize the general construction that Dubrovin and Zhang give in~\cite{DZ-norm}
 to the NLS Frobenius manifold $M_{\rm NLS}$ (which has the potential~\eqref{FNLS27}).
A brief review of Dubrovin--Zhang's construction is in Appendix~\ref{appa}.

Comparing with the notation in Appendix~\ref{appa}, we have for this Frobenius manifold
the flat coordinates ${\bf v}=(v^1,v^2)=(\varphi, \rho)$.
The invariant flat metric, the Euler vector field, and the intersection form of the NLS Frobenius manifold 
read explicitly as follows:
\beq
(\eta^{\alpha\beta}) = \begin{pmatrix} 0 & 1 \\ 1 & 0 \end{pmatrix}, \quad 
E=\varphi\p_\varphi+2 \rho \p_\rho, \quad 
(g^{\alpha\beta}) = \begin{pmatrix}
 2 & \varphi \\
\varphi & 2 \rho \\
\end{pmatrix}.
\eeq
The spectrum of this Frobenius manifold is given by~\cite{CvdLPS, Du93, DZ-norm}
\beq
\mu = \begin{pmatrix} 1/2 & 0 \\ 0 & -1/2 \end{pmatrix}, \quad R=\begin{pmatrix} 0 & 2 \\ 0 & 0 \end{pmatrix}.
\eeq

From the definition given in~\cite{DZ-norm}, one finds the explicit expressions for the Virasoro operators associated to the NLS Frobenius manifold:
\begin{align}
L^{M_{\rm NLS}}_{-1} = \, & \sum_{m\ge1} \tilde t^{1,m} \frac{\p}{\p t^{1,m-1}} + \sum_{m\ge1} \tilde t^{2,m} \frac{\p}{\p t^{2,m-1}} 
+ \frac{\tilde t^{1,0} \tilde t^{2,0}}{\e^2} , 
\label{viranlsm1}
\end{align}
\begin{align}
L^{M_{\rm NLS}}_0 = \, & \sum_{m\ge0} (1+m) \tilde t^{1,m} \frac{\p }{\p t^{1,m}} + 
\sum_{m\ge1} m \tilde t^{2,m} \frac{\p }{\p t^{2,m}} + 2\sum_{m\ge1} \tilde t^{2,m} \frac{\p }{\p t^{1,m-1}}  
+ \frac{\tilde t^{2,0} \tilde t^{2,0}}{\e^2}, \label{viranls0}
\end{align}
\begin{align}
L^{M_{\rm NLS}}_k = \, & \e^2 \sum_{m=1}^{k-1} m! (k-m)! \frac{\p^2}{\p t^{1,m-1} \p t^{1,k-m-1}} \nn\\
& + \sum_{m\ge1} \frac{(k+m)!}{(m-1)!} 
\biggl( \tilde t^{2,m} \frac{\p}{\p t^{2,k+m}} + \tilde t^{1,m-1}\frac{\p}{\p t^{1,k+m-1}} \biggr) \nn\\
 & + 2 \sum_{m\ge0} \alpha_k(m) \tilde t^{2,m} \frac{\p}{\p t^{1,k+m-1}}, \quad k\geq1. \label{viranlsk}
\end{align}
Here, $\tilde t^{\beta,m}=t^{\beta,m}-\delta^{m,1}\delta^{\beta,1}$ ($\beta=1,2$, $m\ge0$), $\alpha_k(0):=k!$, and for $m>0$, $\alpha_k(m):=\frac{(k+m)!}{(m-1)!} \sum_{j=m}^{k+m} \frac 1 j$.

The {\it topological partition function $Z^{M_{\rm NLS}, {\rm top}}({\bf t};\e)=e^{\F^{M_{\rm NLS}, {\rm top}}({\bf t};\e)}$ of the 
NLS Frobenius manifold} $M_{\rm NLS}$ is a particular power series, which 
satisfies the {\it Virasoro constraints} and the dilaton equation:
\begin{align}
& L^{M_{\rm NLS}}_k \bigl(Z^{M_{\rm NLS}, {\rm top}}({\bf t};\e)\bigr) = 0,\quad k\geq-1, \label{viraconstr1012} \\
& \sum_{\alpha=1}^2 \sum_{m\ge0} \tilde t^{\alpha,m} \frac{\p Z^{M_{\rm NLS}, {\rm top}}({\bf t};\e)}{\p t^{\alpha,m}} 
+ \e \frac{\p Z^{M_{\rm NLS}, {\rm top}}({\bf t};\e)}{\p \e}+ \frac{1}{12} Z^{M_{\rm NLS}, {\rm top}}({\bf t};\e) = 0. \label{dilatonnls1210}
\end{align}
Here ${\bf t}=(t^{\alpha,m})_{\alpha=1,\dots,2,\,m\ge0}$, and $\F^{M_{\rm NLS}}({\bf t};\e)$ is called the {\it topological free energy of the 
Frobenius manifold} $M_{\rm NLS}$. For the definition of the topological partition function 
of a semisimple Frobenius manifold see~\cite{DZ-norm} or Appendix~\ref{appa}.
We will give the explicit specialization of Dubrovin--Zhang's construction (see~\cite{DZ-norm} or Appendix~\ref{appa}) 
for the case of the NLS Frobenius manifold. 
Recall that we need to look for solution to~\eqref{viraconstr1012} with $\F^{M_{\rm NLS}, {\rm top}}({\bf t};\e)$ 
having the genus expansion:
\beq
\F^{M_{\rm NLS}, {\rm top}}({\bf t};\e) = \sum_{g\ge0} \e^{2g-2} \F^{M_{\rm NLS}, {\rm top}}_g({\bf t}),
\eeq
where $\F^{M_{\rm NLS}, {\rm top}}_g({\bf t})$ will be called the {\it genus $g$ topological free energy}. 
In particular, we require that $\F^{M_{\rm NLS}}_0({\bf t})$ is given by~\cite{Du96}
\beq\label{genuszeronls1025}
\F^{M_{\rm NLS}, {\rm top}}_0({\bf t}) = \sum_{\alpha,\beta=1}^2\sum_{m_1,m_2\ge0} 
\tilde t^{\alpha,m_1} \tilde t^{\beta, m_2} \Omega_{\alpha,m_1; \beta,m_2}^{M_{\rm NLS},[0]}\bigl({\bf v}_{\rm top}(\bt)\bigr),
\eeq
where $\Omega_{\alpha,m_1; \beta,m_2}^{M_{\rm NLS},[0]}({\bf v})$ are the two-point correlation functions for the principal hierarchy  
of~$M_{\rm NLS}$ (we choose the calibration according to~\cite{CDZ, CvdLPS}), and 
${\bf v}_{\rm top}(\bt)$ denotes the topological solution to the principal hierarchy associated to~$M_{\rm NLS}$.
See~Appendix~\ref{appa} for more details about $\Omega_{\alpha,m_1; \beta,m_2}^{[0]}({\bf v})$ and ${\bf v}_{\rm top}(\bt)$. 
It is shown in~\cite{DZ99} that the $\F^{M_{\rm NLS}, {\rm top}}_0({\bf t})$ constructed by~\eqref{genuszeronls1025} 
satisfies the genus zero part of equations~\eqref{viraconstr1012}.
Denote 
\beq
\Delta\F^{M_{\rm NLS}, {\rm top}}({\bf t};\e) = \sum_{g\ge1} \e^{2g-2} \F^{M_{\rm NLS}, {\rm top}}_g({\bf t}),
\eeq
and consider the following ansatz for the genus $g$ topological free energy $\F_g^{M_{\rm NLS}, {\rm top}}(\bt)$: 
\beq
\F_g^{M_{\rm NLS}, {\rm top}}(\bt) = F_g^{M_{\rm NLS}}\bigl(\varphi_{\rm top}(\bt), \rho_{\rm top}(\bt), \dots, 
\p_{t^{1,0}}^{3g-2} (\varphi_{\rm top}(\bt)), \p_{t^{1,0}}^{3g-2} (\rho_{\rm top}(\bt)) \bigr),
\eeq
where $F_g^{M_{\rm NLS}}(\varphi, \rho, \dots, \varphi_{3g-2}, \rho_{3g-2})$, $g\geq1$, are functions to be determined. 
Let us choose the $\lambda$-periods of the NLS Frobenius manifold:
\beq
p_1({\bf v};\lambda) = \frac12 \sqrt{(\varphi-\lambda)^2-4\rho}, \quad p_2({\bf v};\lambda) = \varphi-\lambda,
\eeq
with the corresponding Gram-matrix 
\beq
(G^{\alpha\beta}) \equiv \begin{pmatrix}
 -2 & 0 \\
 0 & \frac12 \\
\end{pmatrix}.
\eeq
The DZ loop equation for the NLS Frobenius manifold then explicitly reads:
\begin{align}
&  \sum_{r\ge0} \frac{\p \Delta F^{M_{\rm NLS}}}{\p \varphi_r} \p^r\Bigl(\frac{\varphi-\lambda}{D(\lambda)}\Bigr) - 
 2 \sum_{r\ge0} \frac{\p \Delta F^{M_{\rm NLS}}}{\p \rho_r} \p^r\Bigl(\frac{\rho}{D(\lambda)}\Bigr)  \nn\\
& - \frac12  \sum_{r\ge1} \frac{\p \Delta F^{M_{\rm NLS}}}{\p \rho_r} \sum_{k=1}^r 
\binom{r}{k} \p^{k-1}\Bigl(\frac{\varphi-\lambda }{\sqrt{D(\lambda)}}\Bigr) 
\p^{r-k+1} \Bigl(\frac{\varphi-\lambda }{\sqrt{D(\lambda)}}\Bigr) \nn\\
& +  \sum_{r\ge1} \frac{\p \Delta F^{M_{\rm NLS}}}{\p \varphi_r} \sum_{k=1}^r 
\binom{r}{k} \p^{k-1}\Bigl(\frac{\varphi-\lambda }{\sqrt{D(\lambda)}}\Bigr) \p^{r-k+1} 
\Bigl(\frac{1}{\sqrt{D(\lambda)}}\Bigr) \nn\\
& =  - \e^2 \sum_{k,\ell\ge0} 
\biggl(\frac{\p^2 \Delta F^{M_{\rm NLS}}}{\p \varphi_k \p \varphi_\ell} + 
\frac{\p \Delta F^{M_{\rm NLS}}}{\p \varphi_k}\frac{\p \Delta F^{M_{\rm NLS}}}{\p \varphi_\ell}\biggr) 
\p^{k+1}\Bigl(\frac{1}{\sqrt{D(\lambda)}}\Bigr)  \p^{\ell+1} \Bigl(\frac{1}{\sqrt{D(\lambda)}}\Bigr) \nn\\
& + \e^2 \sum_{k,\ell\ge0}  \biggl(\frac{\p^2 \Delta F^{M_{\rm NLS}}}{\p \varphi_k \p \rho_\ell} + 
\frac{\p \Delta F^{M_{\rm NLS}}}{\p \varphi_k}\frac{\p \Delta F^{M_{\rm NLS}}}{\p \rho_\ell}\biggr) 
\p^{k+1}\Bigl(\frac{1}{\sqrt{D(\lambda)}}\Bigr)  \p^{\ell+1} \Bigl(\frac{\varphi-\lambda }{\sqrt{D(\lambda)}}\Bigr) \nn\\
& - \frac{\e^2}4 \sum_{k,\ell\ge0}  \biggl(\frac{\p^2 \Delta F^{M_{\rm NLS}}}{\p \rho_{k} \p \rho_{\ell}} + 
\frac{\p \Delta F^{M_{\rm NLS}}}{\p \rho_k}\frac{\p \Delta F^{M_{\rm NLS}}}{\p \rho_\ell}\biggr) 
\p^{k+1}\Bigl(\frac{\varphi-\lambda }{\sqrt{D(\lambda)}}\Bigr)  \p^{\ell+1} \Bigl(\frac{\varphi-\lambda }{\sqrt{D(\lambda)}}\Bigr) \nn\\
& - \e^2 \sum_{k\ge0} \frac{\p \Delta F^{M_{\rm NLS}}}{\p \varphi_k} \p^{k+1}
\biggl( \frac{ ((\lambda-\varphi)^2+4 \rho)\varphi_1 + 4 (\lambda - \varphi) \rho_1}{D(\lambda)^3} \biggr) \nn\\
& - \e^2 \sum_{k\ge0} \frac{\p \Delta F^{M_{\rm NLS}}}{\p \rho_{k}} \p^{k+1}
\biggl( \frac{4 \rho (\lambda - \varphi)\varphi_1+ ((\lambda-\varphi)^2+4 \rho)\rho_1}{D(\lambda)^3} \biggr) - \frac{\rho}{D(\lambda)^4}, \label{loopnls30}
\end{align}
where $\p=\sum_{k\ge0} \varphi_{k+1} \frac{\p}{\p \varphi_{k}} + \sum_{k\ge0} \rho_{k+1} \frac{\p}{\p \rho_{k}}$, 
$
\Delta F^{M_{\rm NLS}} = \sum_{g\ge1} \e^{2g-2} F^{M_{\rm NLS}}_g$,
and 
\beq
D(\lambda) = (\lambda-\varphi)^2-4 \rho .
\eeq

According to Theorem~B of Appendix~\ref{appa}, the loop equation~\eqref{loopnls30} together with 
\beq\label{dilatonnlsjet}
\sum_{k\ge1} k\varphi_k \frac{\p F_g^{M_{\rm NLS}}}{\p \varphi_k} + \sum_{k\ge1} k\rho_k \frac{\p F_g^{M_{\rm NLS}}}{\p \rho_k}
= (2g-2) F_g^{M_{\rm NLS}} + \frac{\delta_{g,1}}{12} F_g^{M_{\rm NLS}}, \quad g\ge1,
\eeq
has a unique (up to a pure constant for $F_1^{M_{\rm NLS}}$) solution. In this way, we get the 
topological partition function of the NLS Frobenius manifold, that is  
\beq
Z^{M_{\rm NLS}, {\rm top}}({\bf t};\e) = 
\exp \bigl(\e^{-2}\F_0^{M_{\rm NLS}, {\rm top}}({\bf t})
+\Delta F^{M_{\rm NLS}}|_{\varphi_k \mapsto \p_{t^{1,0}}^k (\varphi_{\rm top}(\bt)),\rho_k\mapsto \p_{t^{1,0}}^k (\rho_{\rm top}(\bt)), k\ge0}\bigr).
\eeq

\section{The loop equation for the GUE partition function}\label{section3}
In this section, we derive the loop equation for the GUE partition function $Z(x,{\bf s};\e)$. 

Following~\cite{Du09, DZ-norm, DZ04, Y2} introduce 
\begin{align}
&\sum_{m\geq0} \theta^{M_{\mathbb{P}^1}}_{1,m}(v,u) z^m
:= -2e^{zv}\sum_{j\geq0}\Bigl(\gamma-\frac12 u+\psi(j+1)\Bigr)e^{j u} \frac{z^{2j}}{j!^2}, \label{theta1z}\\
&\sum_{m\geq0} \theta^{M_{\mathbb{P}^1}}_{2,m}(v,u) z^m := z^{-1} \biggl(\sum_{j\geq0} e^{j u+z v} \frac{z^{2j}}{j!^2}-1\biggr), \label{theta2z}
\end{align}
where $\gamma=0.57721\cdots$ is the {\it Euler--Mascheroni constant} and $\psi$ the {\it digamma function}. Denote by 
($v(x,{\bf s}), u(x,{\bf s})$) the unique power-series-in-${\bf s}$ solution to the following equations:
\begin{align}
& \sum_{m\geq0} ((m+1)!s_{m+1}-\delta_{m,1}) \frac{\p \theta^{M_{\mathbb{P}^1}}_{2,m}}{\p v}(v(x,{\bf s}), u(x,{\bf s})) = 0, \label{defvxs1214}\\
& x + \sum_{m\geq0} ((m+1)!s_{m+1}-\delta_{m,1}) \frac{\p \theta^{M_{\mathbb{P}^1}}_{2,m}}{\p u}(v(x,{\bf s}), u(x,{\bf s})) = 0, \label{defuxs1214}
\end{align}
Equivalently speaking, ($v(x,{\bf s}), u(x,{\bf s})$) is the unique power-series-in-${\bf s}$ solution to an infinite family PDEs of hydrodynamic type
\begin{align}
& \frac{\p v}{\p s_m} = m! \p_x \biggl(\frac{\p \theta^{M_{\mathbb{P}^1}}_{2, m}}{\p u}\biggr), \quad 
 \frac{\p u}{\p s_m} = m! \p_x \biggl(\frac{\p \theta^{M_{\mathbb{P}^1}}_{2, m}}{\p v}\biggr) \label{dispersionlesstoda}
\end{align}
with the initial condition $v(x,{\bf 0})=0, u(x,{\bf 0})=\log x$. Here $m\ge1$. According to~\cite{DZ-norm, DZ04} equations~\eqref{dispersionlesstoda} are 
part of the flows of the principal hierarchy (cf.~\eqref{principalh520}) (often called the {\it stationary flows}) associated to the $\mathbb{P}^1$ Frobenius manifold $M_{\mathbb{P}^1}$ whose potential is given by~\eqref{FP11009}. 

Let us define 
\beq \varphi=\p_x \p_{s_1} (\F_0(x, {\bf s})), \quad  \rho=\p_{s_1}^2 (\F_0(x, {\bf s})). \eeq 
Then 
by Theorem~\ref{thmguetoda} and by the definition of the
 tau-structure for the Toda lattice hierarchy~\cite{DY17, Y1, Y2} (see also~\cite{F, mana}) we know that 
\beq\label{mapvuvr}
\varphi(x, {\bf s}) = v(x, {\bf s}), \quad \rho(x, {\bf s}) = e^{u(x,{\bf s})}.
\eeq 
Observe that this is precisely the Legendre-type transformation \cite{Du96} relating $M_{\rm NLS}$ 
and $M_{\mathbb{P}^1}$ that we mentioned in the Introduction. 
The dispersionless limit of the Toda lattice equation
\beq\label{sx1011}
\frac{\p u}{\p y}= \frac{\p v}{\p x}, \quad \frac{\p v}{\p y}= e^u\frac{\p u}{\p x}
\eeq
then yields
\beq
\frac{\p \varphi}{\p y} = \frac{\p \rho}{\p x}, \quad \frac{\p \rho}{\p y} = \rho \frac{\p \varphi}{\p x} \quad 
\Leftrightarrow \quad 
\frac{\p \rho}{\p x} = \frac{\p \varphi}{\p y}, \quad  \frac{\p \varphi}{\p x} = \frac1\rho \frac{\p \rho}{\p y}. \label{xymap35}
\eeq
Here and below, $y=s_1$.
We have the following lemma. 
\begin{lemma}\label{lemma1}
There exist functions 
$$F_g(\varphi, \rho, \varphi_1,\rho_1, \dots, \varphi_{3g-2}, \rho_{3g-2}), \quad g\geq1,$$
such that
\beq
\F_g(x,{\bf s}) = F_g\biggl(\varphi(x,{\bf s}), \rho(x,{\bf s}),\dots, \frac{\p^{3g-2}\varphi(x,{\bf s})}{\p y^{3g-2}} , \frac{\p^{3g-2}\rho(x,{\bf s})}{\p y^{3g-2}}
\biggr) + 
\delta_{g,1}\zeta'(-1). \label{jetguenls}
\eeq
Moreover, for $g=1$, 
\beq
F_1(\varphi, \rho, \varphi_1,\rho_1)=\frac1{24} \log\Bigl(\varphi_1^2- \frac1{\rho} \rho_1^2\Bigr) - \frac1{12} \log \rho,
\eeq
and for $g\ge2$, 
$F_g$ has polynomial dependence in $\varphi_2,\rho_2,\dots,\varphi_{3g-2},\rho_{3g-2}$ and 
has rational dependence in $\varphi_1, \rho_1$ with coefficients depending smoothly in $\varphi, \rho$. 
\end{lemma}
\begin{proof}
By using~\eqref{mapvuvr} and~\eqref{sx1011} iteratively, 
we obtain an invertible map between 
the $x$-derivatives of $u,v$ and the $s$-derivatives of $\varphi, \rho$. 
Then~\eqref{jetguep1} implies~\eqref{jetguenls}. The lemma is proved.
\end{proof}

Following~\cite{DY17}, introduce the loop operator 
\beq
\nabla(\lambda)=\sum_{k\ge 1} \frac{\p}{\p s_{k}} \lambda^{-k-1}
\eeq
\begin{lemma}\label{lemma2}
The following identities hold:
\begin{align}
& \nabla(\lambda) (\F_{0x}) = \frac1{\sqrt{D(\lambda)}}-\frac1{\lambda}, \label{nablaf0x}\\ 
& \nabla(\lambda)^2 (\F_0) = \frac{\sqrt{D(\lambda)}}4 \p_\lambda^2\Bigl(\frac1{\sqrt{D(\lambda)}}\Bigr) + \frac{(\lambda-v)}{2\sqrt{D(\lambda)}}  \p_\lambda\Bigl(\frac1{\sqrt{D(\lambda)}}\Bigr) , \\
& \nabla(\lambda) (\F_{0y}) = \frac{-1+\sqrt{1+4 e^u /D(\lambda)}}{2} = \frac{-1+(\lambda-v)/\sqrt{D(\lambda)}}{2},
\label{nablaf0y}
\end{align}
where $D(\lambda)=(\lambda-v)^2-4e^u=(\lambda-\varphi)^2-4\rho$.
\end{lemma}
\begin{proof}
Based on Theorem~\ref{thmguetoda}, using the tau-structure for the Toda lattice hierarchy given by the matrix-resolvent 
formula~\cite{DY17, Y1, Y2} and taking the dispersionless limit ($\epsilon\to0$), 
we obtain these identities, just as we did at the end of the Section 2.2 of~\cite{Y2}.
\end{proof}

\begin{lemma}\label{lemma1009twopt}
The following formulas hold:
\begin{align}
& \nabla(\lambda) (\varphi) = \p_y\Bigl(\frac1{\sqrt{D(\lambda)}}\Bigr), \\
& \nabla(\lambda) (\rho) = \p_y\Bigl(\frac{\lambda-\varphi}{2\sqrt{D(\lambda)}}\Bigr), \\
& \nabla(\lambda)\Bigl(\frac1{\sqrt{D(\lambda)}}\Bigr) = \frac{((\lambda -\varphi)^2+4\rho) \varphi_y+4(\lambda -\varphi) \rho_y }{\bigl((\lambda -\varphi )^2-4 \rho\bigr)^3} ,\\
& \nabla(\lambda)\Bigl(\frac{\varphi-\lambda}{\sqrt{D(\lambda)}}\Bigr) = \frac{-8\rho(\lambda -\varphi) \varphi_y - (8\rho+2(\lambda -\varphi)^2) \rho_y}{\bigl((\lambda -\varphi )^2-4 \rho\bigr)^3} .
\end{align}
\end{lemma} 
\begin{proof}
Using the definitions of $\varphi, \rho$, $\nabla(\lambda)$ and using \eqref{nablaf0x}, \eqref{nablaf0y}, we 
find the validity of the first two formulas. The third and the fourth formulas then follow from the Leibniz rule. 
\end{proof}

Let us now derive the loop equation for the GUE partition function. 
Let $\Delta F=\sum_{g\ge1} \e^{2g-2} F_g$. 
It is well known that the GUE partition function $Z(x,{\bf s};\e)$ satisfies the dilaton equation 
\beq\label{GUEdilaton1210}
\sum_{j\geq1} \tilde s_j \frac{\p Z(x,{\bf s};\e)}{\p s_{j}} + \e \frac{\p Z(x,{\bf s};\e)}{\p \e}+ \frac{1}{12} Z(x,{\bf s};\e) = 0 
\eeq
and the Virasoro constraints
\beq\label{eqn:vira-GUE-x}
L_k^{\rm GUE} (Z(x, {\bf s}; \e)) = 0 \,, \quad k\geq -1 \,,
\eeq
where the operators $L^{\rm GUE}_k$, $k\geq-1$ are defined by
\begin{align}
L^{\rm GUE}_{-1} = & \sum_{j\geq2} j \tilde s_j \frac{\p}{\p s_{j-1}} 
+\frac{x s_1}{\e^2} , \nn\\
L^{\rm GUE}_0 = & 
\sum_{j\geq1} j \tilde s_j \frac{\p}{\p s_{j}} + \frac{x^2}{\e^2}, \nn\\
L^{\rm GUE}_k = & \e^2 \sum_{j=1}^{k-1} \frac{\p^2}{\p s_j \p s_{k-j}} +  2x \frac{\p}{\p s_k} 
+ \sum_{j\geq1} j \tilde s_j \frac{\p}{\p s_{j+k}}, \quad k\ge1. \nn
\end{align}
Here, $\tilde s_j=s_j-\frac{\delta_{j,2}}2 $.
See e.g.~\cite{AvM, GMMMO91, MMMM91, Mo94}.
It should be noted that the operators $L_k^{\rm NLS}$ (given by~\eqref{viranlsm1}--\eqref{viranlsk}) become $L^{\rm GUE}_k$, $k\ge-1$, under  
$t^{2,m} =  x \delta_{m,0}$, $t^{1,m}= (m+1)! s_{m+1}$, $m\ge0$.

Introduce the operator
\beq
\mu(\lambda)= 
\frac{x}{\lambda}+\sum_{\ell\ge 1} \frac{\ell}2\tilde s_{\ell} \lambda^{\ell-1},
\eeq
and define
\beq
T(\lambda) \;\;=\;\; : \biggl(\e \nabla(\lambda) + \frac{\mu(\lambda)}{\e}\biggr)^2:  \;\;=\;\;  
\e^2 \nabla(\lambda)^2+2\mu(\lambda)\nabla(\lambda) + \frac{\mu(\lambda)^2}{\e^2}.
\eeq
It is straightforward to check that 
\beq
\sum_{m\geq-1} \frac{L_m^{\rm GUE}}{\lambda^{m+2}}  = T(\lambda)_{-}.
\eeq 
Here $(\,)_-$ means taking terms with negative powers in~$\lambda$.
The Virasoro constraints~\eqref{eqn:vira-GUE-x} can then be written equivalently as
\beq\label{idtz25}
T(\lambda)_{-} (Z(x,{\bf s};\e)) = 0.
\eeq
By taking the coefficients of~$\e^{-2}$ on both sides of~\eqref{idtz25} divided by~$Z(x,{\bf s};\e)$ we obtain 
\beq\label{idf025}
2(\mu(\lambda)\nabla(\lambda))_{-} (\F_0) + (\nabla(\lambda)(\F_0) )^2 + \frac{x s_1}{\lambda} + \frac{x^2}{\lambda^2} = 0.
\eeq
By taking the coefficients of the nonnegative powers of~$\e$ on both sides of~\eqref{idtz25} we obtain 
\begin{align}
& 2(\mu(\lambda)\nabla(\lambda))_{-} (\Delta \F)+ 
\e^2 \bigl( (\nabla(\lambda)(\Delta \F))^2+\nabla(\lambda)^2(\Delta \F) \bigr) \nn\\
& \qquad \qquad + 2 \nabla(\lambda)(\F_0) \nabla(\lambda) (\Delta \F) + \nabla(\lambda)^2(\F_0) = 0. \label{highernlsloop528}
\end{align}
Namely, 
\begin{align}
& 2\sum_{k\geq0} ((\mu(\lambda)\nabla(\lambda))_{-}(\rho_k) 
+ \nabla(\lambda)(\F_0) \nabla(\lambda)(\rho_k)) \frac{\p \Delta \F}{\p \rho_k} \nn\\
& \quad + 2\sum_{k\geq0} ((\mu(\lambda)\nabla(\lambda))_{-}(\varphi_k) 
+ \nabla(\lambda)(\F_0) \nabla(\lambda)(\varphi_k)) \frac{\p \Delta \F}{\p \varphi_k}  \nn\\
& \quad + \e^2 \bigl((\nabla(\lambda)(\Delta \F))^2+\nabla(\lambda)^2(\Delta \F) \bigr) 
+ \nabla(\lambda)^2 (\F_0) = 0. \label{namely1010}
\end{align}

Differentiating the identity~\eqref{idf025} with respect to~$y$, we find
\beq
\nabla(\lambda)(\F_0) + 2(\mu(\lambda)\nabla(\lambda))_{-} (\F_{0y}) + 2\nabla(\lambda)(\F_0) \nabla(\lambda)(\F_{0y}) + \frac{x}{\lambda} = 0.
\eeq
Further differentiating this equation with respect to~$x$ and~$y$, respectively, we obtain 
\begin{align}
&\nabla(\lambda)(\F_{0x}) + \frac{2}\lambda \nabla(\lambda)(\F_{0y}) 
+ 2(\mu(\lambda)\nabla(\lambda))_{-} (\F_{0yx}) \nn\\
& \quad + 2\nabla(\lambda)(\F_{0x})\nabla(\lambda)(\F_{0y}) 
+ 2\nabla(\lambda)(\F_0) \nabla(\lambda)(\F_{0yx}) + \frac{1}{\lambda} = 0,\\
& 2 \nabla(\lambda) (\F_{0y}) + 2(\mu(\lambda)\nabla(\lambda))_{-} (\F_{0yy}) 
+ 2 (\nabla(\lambda)(\F_{0y}))^2 + 2\nabla(\lambda)(\F_{0}) \nabla(\lambda)(\F_{0yy}) = 0.
\end{align}
Taking the $y$-derivatives successively in the above two equations we find 
\begin{align}
& \nabla(\lambda)(\F_0) \nabla(\lambda)(\varphi_k) + (\mu(\lambda)\nabla(\lambda))_{-} (\varphi_k)  \nn\\
&  \quad = - \frac12\p_y^{k}\biggl(\frac1{\sqrt{D(\lambda)}} \biggl(k+\frac{\lambda-v}{\sqrt{D(\lambda)}}\biggr)\biggr) \nn\\ 
&  \quad\quad - \sum_{m=1}^{k} \binom{k}{m} \p_y^{m-1} \biggl(\frac{\lambda-v-\sqrt{D(\lambda)}}{2\sqrt{D(\lambda)}}\,\biggr) \p_y^{k-m+1} \Bigl(\frac1{\sqrt{D(\lambda)}}\Bigr), \label{nonlocal11010}\\
 & \nabla(\lambda)(\F_0) \nabla(\lambda)(\rho_k) + (\mu(\lambda)\nabla(\lambda))_{-} (\rho_k) \nn\\
 &  \quad =  - \frac12\p_y^k\biggl(\biggl(\frac{\lambda-v-\sqrt{D(\lambda)}}{2\sqrt{D(\lambda)}} \, \biggr)\biggl(1+k+\frac{\lambda-v}{\sqrt{D(\lambda)}}\biggr)\biggr) \nn\\
 & \quad \quad - \sum_{m=1}^{k} \binom{k}{m} \p_y^{m-1} \biggl(\frac{\lambda-v-\sqrt{D(\lambda)}}{2\sqrt{D(\lambda)}} \,\biggr) \p_y^{k-m+1} \biggl(\frac{\lambda-v}{2\sqrt{D(\lambda)}}\biggr). \label{nonlocal21010}
\end{align}

Using \eqref{namely1010}, \eqref{nonlocal11010}, \eqref{nonlocal21010}, and using the property that 
the transformation from $x, s_1, s_2,\dots$ to $\varphi(x, {\bf s}), \rho(x, {\bf s}), \varphi_y(x, {\bf s}), \rho_y(x, {\bf s}), \dots $ is invertible 
(this interesting property can be proved by first showing with the help of \eqref{defvxs1214}--\eqref{defuxs1214} that 
the transformation from $x, s_1, s_2, s_3, \dots$ to $\varphi(x, {\bf s}), \rho(x, {\bf s}), \varphi_x(x, {\bf s}), \rho_x(x, {\bf s}), \dots $ 
is invertible and then by using~\eqref{xymap35}), 
we obtain 
\begin{align}
& \frac12\sum_{k\geq0} \biggl(
\p_y^k\biggl(1-\Bigl(\frac{\varphi-\lambda}{\sqrt{D(\lambda)}}\Bigr)^2\biggr) - 
 \sum_{m=1}^{k} \binom{k}{m} \p_y^{m-1} \Bigl(\frac{\varphi-\lambda}{\sqrt{D(\lambda)}}\Bigr) 
 \p_y^{k-m+1} \Bigl(\frac{\varphi-\lambda}{\sqrt{D(\lambda)}}\Bigr)\biggr) \frac{\p \Delta F}{\p \rho_k} \nn\\
& + \sum_{k\geq0} \biggl( \p_y^{k} \Bigl(\frac1{\sqrt{D(\lambda)}}
\Bigl(\frac{\varphi-\lambda}{\sqrt{D(\lambda)}}\Bigr)\Bigr) +
 \sum_{m=1}^{k} \binom{k}{m} \p_y^{m-1}\Bigl(\frac{\varphi-\lambda}{\sqrt{D(\lambda)}}\Bigr) \p_y^{k-m+1} \Bigl(\frac1{\sqrt{D(\lambda)}}\Bigr) \biggr) \frac{\p \Delta F}{\p \varphi_k}  \nn\\
&  + \frac{\e^2}4 \sum_{k_1,k_2\geq0} \p_y^{k_1+1}\Bigl(\frac{\varphi-\lambda}{\sqrt{D(\lambda)}}\Bigr) \p_y^{k_2+1}
\Bigl(\frac{\varphi-\lambda}{\sqrt{D(\lambda)}}\Bigr)
\biggl(\frac{\p^2 \Delta F}{\p \rho_{k_1} \p \rho_{k_2}} 
+ \frac{\p \Delta F}{\p \rho_{k_1}} \frac{\p \Delta F}{\p \rho_{k_2}} \biggr)\nn\\
&  + \e^2 \sum_{k_1,k_2\geq0} \p_y^{k_1+1} \Bigl(\frac1{\sqrt{D(\lambda)}}\Bigr) \p_y^{k_2+1}\Bigl(\frac1{\sqrt{D(\lambda)}}\Bigr)
\biggl(\frac{\p^2 \Delta F}{\p \varphi_{k_1} \p \varphi_{k_2}} + \frac{\p \Delta F}{\p \varphi_{k_1}} \frac{\p \Delta F}{\p \varphi_{k_2}} \biggr)  \nn\\
&  - \e^2 \sum_{k_1,k_2\geq0} \p_y^{k_1+1}\Bigl(\frac{\varphi-\lambda}{\sqrt{D(\lambda)}}\Bigr) \p_y^{k_2+1}\Bigl(\frac1{\sqrt{D(\lambda)}}\Bigr) 
\biggl(\frac{\p^2 \Delta F}{\p \rho_{k_1} \p \varphi_{k_2}} + \frac{\p \Delta F}{\p \rho_{k_1}} \frac{\p \Delta F}{\p \varphi_{k_2}} \biggr) \nn\\
& + \e^2 \sum_{k\ge0} \biggl(\p_y^{k+1}\Bigl(\nabla(\lambda)\Bigl(\frac1{\sqrt{D(\lambda)}}\Bigr)\Bigr) \frac{\p \Delta F}{\p \varphi_k} - \frac12 \p_y^{k+1}\Bigl(\nabla(\lambda)\Bigl(\frac{\varphi-\lambda}{\sqrt{D(\lambda)}}\Bigr)\Bigr) \frac{\p \Delta F}{\p \rho_k} \biggr) \nn\\
& + \nabla(\lambda)^2 (\F_0) = 0. \label{gueloop1214}
\end{align}
By using Lemmas \ref{lemma2}, \ref{lemma1009twopt}, we find that this loop equation 
coincides with~\eqref{loopnls30}. It would be interesting to make a connection between 
the loop equation~\eqref{gueloop1214} and the results in~\cite{Er, EMP}.

It follows from the dilaton equation~\eqref{GUEdilaton1210} that 
\beq\label{dilatonguenlsjet}
\sum_{k\ge1} k\varphi_k \frac{\p F_g}{\p \varphi_k} + \sum_{k\ge1} k\rho_k \frac{\p F_g}{\p \rho_k}
= (2g-2) F_g + \frac{\delta_{g,1}}{12} F_g, \quad g\ge1,
\eeq
which coincides with~\eqref{dilatonnlsjet}. 

We therefore conclude from the uniqueness of the solution to \eqref{gueloop1214}, \eqref{dilatonguenlsjet} that 
 $\Delta F$ can only differ from $\Delta F^{M_{\rm NLS}}$ by a constant that is independent of~$\e$. Let us 
 choose the constant for $F_1^{M_{\rm NLS}}$ so that $F_1^{M_{\rm NLS}}=F_1$. 

From~\eqref{viraconstr1012} it is easy to check that 
$\mathcal{F}^{M_{\rm NLS}, {\rm top}}_0({\bf t};\e) \big|_{t^{1,m}=(m+1)!s_{m+1}, \, t^{2,m} = x \delta^{m,0}, \, m\ge0}$ satisfies the genus zero part of the Virasoro constraints 
for the GUE free energy as well as the genus zero dilaton equation. By a straightforward computation, we also have 
$$\mathcal{F}^{M_{\rm NLS}, {\rm top}}_0({\bf t};\e) \big|_{t^{1,m}=0, \, t^{2,m} = x \delta^{m,0}, \, m\ge0} = \mathcal{F}_0(x,{\bf 0}).$$ 
Therefore, $\mathcal{F}^{M_{\rm NLS}, {\rm top}}_0({\bf t};\e) \big|_{t^{1,m}=(m+1)!s_{m+1}, \, t^{2,m} = x \delta^{m,0}, \, m\ge0}=\F_0(x,{\bf s})$. 
In particular, 
$$\varphi(x,{\bf s})= \varphi_{\rm top}(\bt)|_{t^{1,m}=(m+1)!s_{m+1},\, t^{2,m} = x \delta^{m,0} , \, m\ge0}, \quad \rho(x,{\bf s})= \rho_{\rm top}(\bt)|_{t^{1,m}=(m+1)!s_{m+1}, \, t^{2,m} = x \delta^{m,0}, \, m\ge0}.$$

Hence we arrive at the following theorem, originally observed by Dubrovin~\cite{Du09} (a result very similar to this theorem was 
also proved in~\cite{CvdLPS} by using a different method). 

\begin{theorem}\label{thmmain} The following identity holds:
\beq\label{ZZM}
Z(x,{\bf s};\e) = Z^{M_{\rm NLS}, {\rm top}}\big|_{t^{1,m}=(m+1)!s_{m+1}, \, t^{2,m} = x \delta^{m,0}}.
\eeq
\end{theorem}

We see that starting with the $(3g-2)$ jet representation~\eqref{jetguenls} for the 
genus $g$ GUE free energy $\F_g(x,{\bf s})$, $g\ge1$, we can derive the loop equation which 
coincides with the DZ loop equation for the NLS Frobenius manifold. 
We remark that, in a similar way, 
if we start with the $(3g-2)$ jet representation~\eqref{jetguep1} for the 
genus $g$ GUE free energy $\F_g(x,{\bf s})$, $g\ge1$,
then we can obtain by solving the Virasoro constraints~\eqref{eqn:vira-GUE-x} the loop equation 
which will coincide with the DZ loop equation 
for the $\mathbb{P}^1$ Frobenius manifold. We omit this detail as it is almost identical with the details that we have given. 
We expect that, for a semisimple Frobenius manifold, the similar situation happens when there is a certain resonance in its spectrum. 
We study this in subsequent publications. 

\section{Concluding remarks}\label{section4}
In~\cite{Y2} we summarized several developments in topological gravity and matrix gravity 
for higher genera through the jet-space. Let us present in the following diagram 
a more global picture on 
the relations among a few interesting partition functions in the {\it big phase space}:
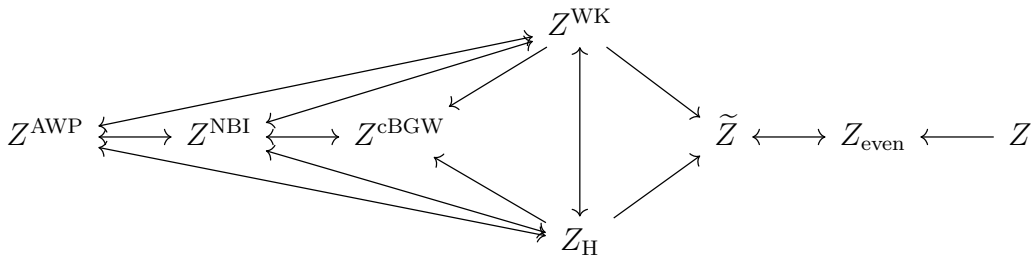
\begin{center}
\begin{tikzcd}
& & & Z^{\rm WK}  \arrow[leftrightarrow]{dd} \arrow[dr] \arrow[dl] \arrow[leftrightarrow]{dll} \arrow[leftrightarrow]{dlll}   & & & \\ 
Z^{\rm AWP} & Z^{\rm NBI} \arrow[leftrightarrow]{l} & Z^{\rm cBGW} \arrow[leftrightarrow]{l}  &  & \widetilde Z & Z_{\rm even} \arrow[leftrightarrow]{l}  & Z \arrow[l] \\
& & & Z_{{\rm H}} \arrow[ur] \arrow[ul] \arrow[leftrightarrow]{ull} \arrow[leftrightarrow]{ulll} &  &  & \\
\end{tikzcd}
\end{center}
Here, $Z_{\rm even}=Z_{\rm even}(x,s_2,s_4,\dots;\e)$ denotes the GUE partition function with even couplings (defined by 
restricting $s_1=s_3=s_5=\dots=0$ in the GUE partition function $Z(x,{\bf s};\e)$;  for the significance of $Z_{\rm even}$ in 
quantum gravity see e.g.~\cite{Witten}), 
$\widetilde Z=\widetilde Z(x,s_2,s_4,\dots;\e)$ denotes the modified GUE partition function (which can be obtained from 
$Z_{\rm even}(x,s_2,s_4,\dots;\e)$ via an invertible product formula~\cite{DLYZ2, DY17-2, DY20} giving the first arrow from right of the above diagram), 
$Z^{\rm WK}$ denotes 
the celebrated Witten--Kontsevich partition function~\cite{Kontsevich, Witten, Y2} (topological gravity after Witten),
 $Z_{\rm H}$ denotes a very special cubic Hodge partition function~\cite{DLYZ1, DLYZ2, DY17-2},  
 $Z_{\rm cBGW}$ denotes the corrected generalized BGW partition function (see e.g.~\cite{Ale18, MMS, YZhang}; BGW 
 stands for Br\'ezin--Gross--Witten), $Z_{\rm NBI}$ denotes the NBI partition 
 function~\cite{AC, XY} (NBI stands for non-abelian Born--Infeld), and $Z_{\rm AWP}$ denotes the 
 partition function of certain mixed $\psi$ and $\kappa$ classes introduced in~\cite{Ale21, YZagier}.
 The arrow from $Z_{\rm H}$ to $\widetilde Z$ is the Hodge-GUE correspondence~\cite{DLYZ2, DY17-2}; 
  the arrow between $Z_{\rm H}$ and $Z_{\rm WK}$ is the Hodge-WK correspondence~\cite{YZagier} (cf.~\cite{Ale21});
  the arrow from $Z^{\rm WK}$ to $\widetilde Z$ is called the WK-GUE correspondence~\cite{YZagier};
 the arrow from $Z^{\rm WK}$ to $Z_{\rm cBGW}$ is given by the WK-BGW correspondence~\cite{YZagier};   
 the arrow between $Z^{\rm WK}$ and $Z^{\rm NBI}$ is given by shifting time variables~\cite{AC, XY}; 
  the arrow between $Z^{\rm WK}$ and $Z^{\rm AWP}$ is given by shifting time variables~\cite{Ale21, YZagier} (cf.~\cite{KMZ, LX, MZ}); 
   the arrow between $Z_{\rm H}$ and $Z^{\rm AWP}$ comes from a result given in~\cite{Ale21, YZagier}; 
   the arrow between $Z_{\rm H}$ and $Z^{\rm NBI}$ can be given as a result of the Hodge--WK correspondence and the 
   definition of $Z^{\rm NBI}$ (cf.~\cite{XY}); 
 the arrow from  $Z_{\rm H}$ to $Z^{\rm cBGW}$ is the Hodge-BGW correspondence~\cite{YZhang}; 
  the arrow between $Z^{\rm AWP}$ and $Z^{\rm NBI}$ is given again by shifting time variables;
 the arrow between  $Z^{\rm NBI}$  and $Z^{\rm cBGW}$ can~\cite{XY} be given by 
  the Galilean symmetry of the KdV hierarchy.

\begin{appendix}
\section{Topological partition function of a semisimple Frobenius manifold}\label{appa}
In this appendix, we review the topological partition function of a semisimple Frobenius manifold and the Dubrovin--Zhang (DZ) loop equation. 
For more about the theory of Frobenius manifolds we refer to~\cite{Du96, Du99, DZ-norm}. 

Let $(M, \, \cdot, \, e, \, \langle\,,\,\rangle, E)$ be an $n$-dimensional Frobenius manifold. 
Take ${\bf v}=(v^1,\dots,v^n)$ a system of flat coordinates on~$M$ with respect to $\langle\,,\,\rangle$.
Denote 
\beq
\eta_{\alpha\beta} = \langle\p_\alpha, \p_\beta\rangle, \quad \eta=(\eta_{\alpha\beta}), \quad 
c_{\alpha\beta\rho} = \langle\p_\alpha\cdot\p_\beta, \p_\rho\rangle, \quad 
\alpha,\beta,\rho = 1,\dots,n.
\eeq 
Here and below we use $\p_\alpha:=\p/\p v^\alpha$, $\alpha=1,\dots,n$, to denote the coordinate vector fields. 
Define $(\eta^{\alpha\beta})=\eta^{-1}$. 
We assume that $E$ is diagonalizable, and 
 choose the flat coordinates $v^1,\dots,v^n$, 
such that the vector fields $e$ and $E$ have the form
\begin{align}
&e=\p_1,  \label{e517} \\
&E=\sum_{\beta=1}^n E^\beta \p_\beta, \quad E^\beta
= \Bigl(1-\frac{d}2-\mu_\beta\Bigr) v^\beta + r^\beta, \label{E517} 
\end{align}
where $\mu$'s and $r$'s are constants with 
$\mu_1=-d/2$, $r^1=0$. Denote $\mu:={\rm diag}(\mu_1,\dots,\mu_n)$,  
and $v_\alpha:=\sum_{\rho=1}^n \eta_{\alpha\rho} v^\rho$, $\alpha=1,\dots,n$, and  
\beq
c^{\alpha}_{\beta\sigma} := \sum_{\rho=1}^n \eta^{\alpha\rho} c_{\rho \beta \sigma}, \quad \alpha,\beta,\sigma = 1,\dots,n.
\eeq
The intersection form $(\,,\,):T^*M\times T^*M\rightarrow \CC$ on the Frobenius manifold~$M$ is defined by 
\beq
(a,b) := i_E(a\cdot b), \quad a,b\in T^*M.
\eeq
Denote $g^{\alpha\beta}=(dv^\alpha,dv^\beta)$. We know from e.g.~\cite{Du96} that $g^{\alpha\beta}$ is a 
contravariant flat metric on~$M$.

In~\cite{Du93, Du96} Dubrovin introduced a one-parameter family of 
affine connections $\widetilde\nabla(z)$ by
\beq
\widetilde \nabla(z)_{X} Y := \nabla_X Y + z X \cdot Y, \quad \forall \, X,Y\in \mathcal{X}(M), \, z\in \CC,
\eeq
and showed that for any $z\in \CC$ the connection $\widetilde \nabla(z)$ is flat.
This one-parameter family of flat connections $\widetilde \nabla(z)$ is called the 
 {\it deformed flat connection}~\cite{Du96, DZ-norm}, also known as the {\it Dubrovin connection}~\cite{MS}. 
 It was shown in~\cite{Du96} (cf.~\cite{Du99}) 
that the deformed flat connection can be extended to a flat affine connection on $M\times \CC^*$ through 
\beq
\widetilde \nabla_{\p_z} X := \frac{\p X}{\p z} + E \cdot X - \frac1 z \mathcal{V} X, \quad 
\widetilde \nabla_{\p_z} \p_z := 0,  \quad  \widetilde \nabla_X \p_z := 0
\eeq
for $X$ being an arbitrary holomorphic vector field on $M \times \mathbb{C}^*$ with 
zero component along~$\p_z$. Here, $\p_z:=\p/\p z$, and 
$\mathcal{V}:=\frac{2-d}2 {\rm id} - \nabla E$.
The affine connection~$\widetilde\nabla$ is called the {\it extended deformed flat connection} on $M\times\CC^*$. 

A holomorphic function $f$ on some open subset of $M\times \CC^*$ is called {\it $\widetilde \nabla$-flat}, if 
\beq
\widetilde \nabla \biggl(\,\sum_{\alpha=1}^n \p_\alpha (f) dv^\alpha\biggr)=0.
\eeq
Denote 
\beq
\mathcal{U}^\alpha_\beta({\bf v}) := \sum_{\rho=1}^n E^\rho({\bf v}) c^\alpha_{\rho\beta}({\bf v}).
\eeq
Then for a {\it $\widetilde \nabla$-flat} holomorphic function~$f$, 
the differential system for its gradient $y^\alpha=\sum_{\beta=1}^n \eta^{\alpha\beta} \p_\beta(f)$ 
reads 
\begin{align}
&\p_\alpha (y) = z C_\alpha y, \label{pdedub517}\\
&\p_z(y) = \Bigl(\mathcal{U}+\frac{\mu}{z} \Bigr) y, \label{odedub517}
\end{align}
where $y=(y^1,\dots,y^n)^T$, $C_\alpha=(C_{\alpha\beta}^\rho)$, $\alpha=1,\dots,n$.
It is shown in~\cite{Du99} (cf. \cite{Du96,DZ-norm}) that for a sufficiently small ball $B\subset M$, 
there exists a fundamental solution 
matrix to \eqref{pdedub517}--\eqref{odedub517} of the form
\beq\label{mY517}
\mathcal{Y}= \Theta({\bf v};z) z^\mu z^R,
\eeq
such that $\Theta({\bf v};z)$ is analytic on $B\times \CC$, and 
\beq
\Theta({\bf v};0) = I, \quad \eta^{-1} \Theta({\bf v};-z)^T \eta \Theta({\bf v};z)= I.
\eeq
Here, $R$ is a certain constant matrix
 satisfying 
\begin{align}
& R=\sum_{s\ge1} R_s, \quad z^\mu R z^{-\mu} = \sum_{s\ge1} R_s z^s, \label{Rmu512}\\
& \eta^{-1} R_s^T \eta = (-1)^{s+1} R_s, \quad s\ge1.
\end{align}
Here the two summations in~\eqref{Rmu512} are actually finite sums. From~\eqref{Rmu512} we know that 
\beq
(R_s)^\alpha_\beta \neq 0 \quad \mbox{only if } \mu_\alpha-\mu_\beta =s, \quad  s\ge1.
\eeq 
The matrices $\mu,R$ are called 
 the monodromy data~\cite{Du96, Du99} at $z=0$ of the Frobenius manifold.

The fundamental solution matrix~\eqref{mY517} gives the existence of $\widetilde \nabla$-flat holomorphic functions 
$\tilde{v}_1({\bf v};z),\dots,\tilde{v}_n({\bf v};z)$ on $B \times \CC^*$ of the form
\begin{align}
(\tilde{v}_1({\bf v};z),\dots,\tilde{v}_n({\bf v};z)) = (\theta_1({\bf v};z),\dots,\theta_n({\bf v};z)) z^{\mu} z^R,
\end{align}
where $\theta_\alpha({\bf v};z)=:\sum_{m\geq0} \theta_{\alpha,m}({\bf v})z^m$, $\alpha=1,\dots,n$,  
are analytic functions on $B\times \CC$.
Together with FM1, one can verify that 
the coefficients $\theta_{\alpha,m}=\theta_{\alpha,m}({\bf v})$,
$\alpha=1,\dots,n$, $m\geq0$, can be chosen 
to satisfy the following equations: 
\begin{align}
&
\theta_{\alpha,0}=v_\alpha, \label{theta000511}\\
&\p_1(\theta_{\alpha,m+1}) = \theta_{\alpha,m}, \quad  m\geq0, \label{thetacali512}\\
&
\p_{\alpha} \p_{\beta} (\theta_{\rho,m+1}) = \sum_{\sigma=1}^n c^\sigma_{\alpha\beta} 
\p_{\sigma} (\theta_{\rho,m}), \label{theta1c}\\
&\langle \nabla \theta_\alpha({\bf v};z),  \nabla \theta_\beta({\bf v};-z) \rangle = \eta_{\alpha\beta}, \label{orthogtheta511} \\
& E (\p_{\beta}(\theta_{\alpha,m})) = 
(m+\mu_\alpha+\mu_\beta) \, \p_{\beta}(\theta_{\alpha,m}) + 
\sum_{\sigma=1}^n \sum_{k=1}^m (R_k)^\sigma_\alpha \p_{\beta} (\theta_{\sigma,m-k}) , \quad m\geq 0. \label{theta2c} 
\end{align}
Here $\alpha,\beta,\rho=1,\dots,n$.
We call a choice of $(\theta_{\alpha,m})_{\alpha=1,\dots,n,\,m\geq 0}$ 
satisfying \eqref{theta000511}--\eqref{theta2c} 
a {\it calibration} on the Frobenius manifold~$M$. 
A Frobenius manifold with a chosen calibration is called {\it calibrated}.
We call $\tilde{v}_1({\bf v};z),\dots,\tilde{v}_n({\bf v};z)$ 
the {\it deformed flat coordinates}.

From now on $M$ is a calibrated Frobenius manifold.
Following~\cite{Du93, Du96}, define the {\it two-point correlation functions for the principal hierarchy}
$\Omega_{\alpha,m_1;\beta,m_2}^{[0]}$, $\alpha,\beta=1,\dots,n, \, m_1,m_2\geq0$, 
by means of generating series as follows:
\beq\label{deftwopoint511}
\sum_{m_1,m_2\ge0} \sum_{\alpha,\beta=1}^n \Omega_{\alpha,m_1;\beta,m_2}^{[0]}({\bf v}) z^{m_1} w^{m_2} = 
\frac{\langle \nabla\theta_\alpha({\bf v};z), \nabla\theta_\beta({\bf v};w) \rangle-\eta_{\alpha\beta}}{z+w} .
\eeq
Define
\beq
F^M({\bf v})=\frac12 \Omega_{1,1;1,1}^{[0]}({\bf v}) -\sum_{\rho=1}^n v^\rho \Omega_{\rho,0;1,1}^{[0]}({\bf v}) 
+ \frac12 \sum_{\alpha,\beta=1}^n 
v^\alpha v^\beta \Omega_{\alpha,0;\beta,0}^{[0]}({\bf v}).
\eeq
We call $F^M({\bf v})$ the {\it potential} of~$M$, which satisfies 
\beq\label{dddF519}
\frac{\p^3 F^M({\bf v})}{\p v^\alpha\p v^\beta \p v^\rho} = c_{\alpha\beta\rho}({\bf v}), \quad \alpha,\beta,\rho=1,\dots,n.
\eeq

Let us continue to recall the {\it principal hierarchy} associated to~$M$, which is a hierarchy of systems of PDEs of 
hydrodynamic type, introduced by Dubrovin~\cite{Du93, Du96}: 
\beq\label{principalh520}
\frac{\p v^\alpha}{\p t^{\beta,m}} = \sum_{\gamma=1}^n \eta^{\alpha\gamma} \p_X \Bigl(\frac{\p \theta_{\beta,m+1}}{\p v^\gamma}\Bigr) 
= \sum_{\gamma,\rho,\sigma=1}^n \eta^{\alpha\gamma} c_{\gamma\sigma}^\rho 
\frac{\p \theta_{\beta,m}}{\p v^\rho} \frac{\p v^\sigma}{\p X}, \quad \alpha,\beta=1,\dots,n,\, m\ge0.
\eeq
Dubrovin~\cite{Du93, Du96} showed that the principal hierarchy is integrable, i.e., the flows in this hierarchy 
pairwise commute. It is easy to see that the $\p_{t^{1,0}}$-flow reads: 
\beq
\frac{\p v^\alpha}{\p t^{1,0}} = \frac{\p v^\alpha}{\p X}, \quad \alpha=1,\dots,n.
\eeq
So we identify $t^{1,0}$ with the space variable~$X$.
By a {\it power-series solution} ${\bf v}={\bf v}(\bt)$ to the principal hierarchy~\eqref{principalh520},  we mean 
an element ${\bf v}(\bt)$ in $\bigl(\CC[[t^{1,0}-t^1_o,\dots,t^{n,0}-t^n_o]][[{\bf t}_{>0}]]\bigr)^n$ satisfying all 
equations in~\eqref{principalh520}. Here, ${\bf t}=(t^{\beta,m})_{\beta=1,\dots,n,\, m\ge0}$, 
${\bf t}_{>0}=(t^{\beta,m})_{\beta=1,\dots,n,\, m\ge1}$, and 
$t^\beta_o$, $\beta =1,\dots,n$, are certain constants which we call the {\it primary initials}. 
We call $t^{1,0},\dots,t^{n,0}$ the {\it primary times}, 
and $t^{\beta,m}$, $\beta=1,\dots,n$, $m\ge1$, the {\it descendent times}.
We note that  $(v^\alpha(t^1_o,\dots,t^n_o,0,\dots,0))_{\alpha=1}^n$ should belong to the small ball~$B$.
The power-series solutions to~\eqref{principalh520} can be uniquely 
characterized by their initial data  
\beq
f^\alpha(x)=v^\alpha({\bf t})|_{t^{\beta,0}=\delta^{\beta,1} (X-t^1_o) + t^\beta_o, \, {\bf t}_{>0}={\bf 0}, \, \beta=1,\dots,n}, \quad \alpha=1,\dots,n.
\eeq

For a power-series solution ${\bf v}={\bf v}({\bf t})=(v^1({\bf t}),\dots, v^n({\bf t}))$ 
to the principal hierarchy,  
there exists a function $\F_0=\F_0({\bf t})\in\CC[[t^{1,0}-t^1_o,\dots,t^{n,0}-t^n_o]][[{\bf t}_{>0}]]$, such that 
\beq\label{defif01025}
\frac{\p^2 \log \F_0}{\p t^{\alpha,m_1} \p t^{\beta,m_2}} = \Omega_{\alpha,m_1; \beta,m_2}^{[0]}({\bf v}({\bf t})), \quad 
\alpha,\beta=1,\dots,n,\, m_1,m_2\ge0.
\eeq
We call 
$\tau^{[0]}:=e^{\e^{-2} \F_0}$ the {\it tau-function of the solution}~${\bf v}({\bf t})$ to the principal hierarchy~\eqref{principalh520}.
We call the map 
\beq
(t^{1,0},\dots,t^{n,0}) \mapsto \bigl(V^1,\dots,V^n\bigr)\in B,
\eeq
defined by 
$V^\alpha(t^{1,0},\dots,t^{n,0})= v^\alpha({\bf t})|_{{\bf t}_{>0}={\bf 0}}$
the {\it Riemann map associated to the solution}~${\bf v}({\bf t})$.

The so-called {\it topological solution} ${\bf v}_{\rm top}(\bf t)$ to the principal hierarchy is defined as the unique 
solution in $(\CC[[t^{1,0}-t^1_o,\dots,t^{n,0}-t^n_o]][[{\bf t}_{>0}]])^n$ to~\eqref{principalh520} specified by 
\beq\label{initopprin521}
f_{\rm top}^\alpha(x):= v_{\rm top}^\alpha({\bf t})|_{t^{\beta,0}=\delta^{\beta}_{1} (X-t^1_o) + t^\beta_o, \, {\bf t}_{>0}={\bf 0}, \, \beta=1,\dots,n} = \delta^{\alpha,1} (X-t^1_o) + t^\alpha_o, \quad \alpha=1,\dots,n. 
\eeq
It is easy to verify that ${\bf v}_{\rm top}(\bf t)$ has the property:
\beq\label{vtopini0521}
v_{\rm top}^\alpha({\bf t})|_{{\bf t}_{>0}={\bf 0}} \equiv t^{\alpha,0}, \quad \alpha=1,\dots,n.
\eeq
That is, 
the Riemann map associated to the topological solution is the identity map. 
Following Dubrovin~\cite{Du96}, define the genus zero topological free energy $\F^M_0(\bt)$ by
\beq\label{fm01213}
\F^{\rm top}_0(\bt) = \frac12 \sum_{\alpha,\beta=1}^n \sum_{m_1,m_2} \tilde t^{\alpha,m_1} \tilde t^{\beta,m_2} 
\Omega_{\alpha,m_1;\beta,m_2}^{[0]} ({\bf v}_{\rm top}(\bt)),
\eeq
where $\tilde t^{\alpha,k}=t^{\alpha,k}-\delta^{\alpha,1}\delta^{k,1}$.
It can be verified that $\F^{\rm top}_0(\bt)$ satisfies~\eqref{defif01025} with ${\bf v}$ being ${\bf v}_{\rm top}$. So
the exponential $e^{\e^{-2} \F^{\rm top}_0(\bt)}$ is a particular tau-function for the principal hierarchy, 
which is called the {\it topological tau-function} for the principal hierarchy.

We proceed and consider the Virasoro constraints~\cite{DVV, DZ99, DZ-norm, EHX, Getzler, Gi01, LT} following Dubrovin and Zhang~\cite{DZ-norm}. 
Define the linear operators $a_{\alpha,p}$, $p\in \mathbb{Z}+1/2$, by
\beq\label{creation_annihilation}
a_{\alpha,p} = \left\{\begin{array} {cl}
\epsilon \frac{\p}{\p
t^{\alpha,p-1/2}}, & p>0, \\
\\
\epsilon^{-1} (-1)^{p+1/2} \eta_{\alpha\beta} \,
\tilde t^{\beta,-p-1/2}, &  p<0. \end{array}\right.
\eeq
Denote
\begin{align}
f_\alpha(\lambda;\nu) &= \int_0^\infty \frac{dz}{z^{1-\nu}} e^{-\lambda z}
\sum_{p\in\mathbb{Z}+\frac12} \sum_{\beta=1}^n a_{\beta,p} \left(z^{p+\mu} z^R\right)^\beta_\alpha, \label{natural_with_nu}\\
G^{\alpha\beta}(\nu) &= -\frac{1}{2\pi} \Bigl[\bigl(e^{\pi i R}e^{\pi i (\mu+\nu)}+e^{-\pi i R}e^{-\pi i (\mu+\nu)}\bigr)\eta^{-1}\Bigr]^{\alpha\beta}.\label{Galphabetanu0411}
\end{align}
Here, the right-hand side of~\eqref{natural_with_nu} is understood as a term-by-term integral.
Following~\cite{DZ-norm} define the regularized stress tensor $T(\lambda;\nu)$ as follows:
\begin{equation}
T(\lambda;\nu) = -\frac12 \sum_{\alpha,\beta=1}^n: \p_\lambda(f_\alpha(\lambda;\nu)) G^{\alpha\beta}(\nu) \p_\lambda(f_\beta(\lambda;-\nu)):
+ \frac{1}{4\lambda^2} {\rm tr} \Bigl(\frac{1}{4}-\mu^2\Bigr). \label{Virasoro_field_with_nu}
\end{equation}
Here, ``$:~:$" denotes the {\it normal
ordering} (putting the annihilation operators always on the right of the creation operators).
The regularized stress tensor has the form~\cite{DZ-norm}:
\begin{align}\label{TnuLmnu427}
& T(\lambda;\nu) = \sum_{k\in \mathbb{Z}} \frac{L_k(\nu)}{\lambda^{k+2}}.
\end{align}
Then, as it is shown in~\cite{DZ-norm} (see also~\cite{DZ04}), the operators $L_k$, $k\ge-1$, defined by 
\beq\label{viralm0922}
L_k:=\lim_{\nu\to 0} L_k(\nu), \quad k\geq-1,
\eeq 
yield the {\it Virasoro operators} given in~\cite{DZ99, DZ-norm} (cf.~\cite{EHX, LT}) for the Frobenius manifold~$M$.
 The operators $L_k$ satisfy the following Virasoro commutation relations:
\beq\label{Lmn0407}
[L_k,L_\ell] = (k-\ell) L_{k+\ell}, \quad \forall\,k,\ell\geq -1.
\eeq
In~\cite{DZ99} Dubrovin and Zhang proved the following genus zero Virasoro constraints:
\beq
e^{-\e^{-2} \F^{\rm top}_0(\bt)} L_k \bigl(e^{\e^{-2} \F^{\rm top}_0(\bt)}\bigr) = \mathcal{O}(1) \quad (\e\to0), \quad k\geq-1.
\eeq

Under the semisimplicity assumption, Dubrovin and Zhang~\cite{DZ-norm} consider the reconstruction of  
higher genus free energies via solving  
 the following Virasoro constraints together with the dilaton equation:
\begin{align}
& L_m(Z^{\rm top}(\bt;\e)) = 0, \quad \forall\, m\geq-1, \label{virazM1213}\\
& \sum_{\alpha=1}^n \sum_k \tilde t^{\alpha,k} \frac{\p Z^{\rm top}(\bt;\e)}{\p t^{\alpha,k}} 
+ \e \frac{\p Z^{\rm top}(\bt;\e)}{\p \e}+ \frac{n}{24} Z^{\rm top}(\bt;\e) = 0, 
\label{dilatonzM1213}\\
& Z^{\rm top}(\bt;\e) = e^{\sum_{g\ge0} }\e^{2g-2} \F^{\rm top}_g(\bt).
\end{align}
Here $\F^{\rm top}_g(\bt)$ for $g\geq1$ will be called the {\it genus $g$ topological free energy of~$M$} and $Z^{\rm top}(\bt;\e)$ will 
be called the {\it topological partition function of~$M$} (see also~\cite{Givental1, Gi01}).
The way by Dubrovin and Zhang is to consider the following ansatz \cite{DW, DZ-norm, EYY, Getzler2, Witten} (cf.~also~\cite{IZ}) for $\F^{\rm top}_g$:
\begin{align}
&  \F^{\rm top}_g(\bt) 
= F_g^M\Bigl({\bf v}_{\rm top}(\bt),\frac{\p {\bf v}_{\rm top}(\bt)}{\p X}, \dots, \frac{\p^{3g-2} {\bf v}_{\rm top}(\bt)}{\p X^{3g-2}}\Bigr),
\end{align}
and to recast the Virasoro constraints into the {\it loop equation} 
\begin{align}
& \sum_{\rho=1}^n \sum_{r\ge0} \frac{\p \Delta F^M}{\p v^{\rho}_r} \p^r\biggl(\biggl(\frac1{E-\lambda}\biggr)^\rho\,\biggr) 
+ \sum_{\rho,\alpha,\beta=1}^n 
\sum_{r\ge1} \frac{\p \Delta F^M}{\p v^{\rho}_r} \sum_{k=1}^r 
\binom{r}{k} \p^{k-1}\Bigl(\frac{\p p_\alpha}{\p v^1}\Bigr) G^{\alpha\beta} \p^{r-k+1} \Bigl(\frac{\p p_\beta}{\p v_\rho}\Bigr) \nn \\
& = - \frac1{16} {\rm tr} \bigl((\mathcal{U}-\lambda)^{-2}\bigr) + \frac14 {\rm tr} \bigl(((\mathcal{U}-\lambda)^{-1} \mathcal{V})^2\bigr) \nn\\
& \quad + \frac{\e^2}2 \sum_{\alpha,\beta,\rho,\sigma=1}^n\sum_{k,\ell\ge0}  \biggl(\frac{\p^2 \Delta F^M}{\p v^{\rho}_k \p v^{\sigma}_\ell} + 
\frac{\p \Delta F^M}{\p v^{\rho}_k }\frac{\p \Delta F^M}{\p v^{\sigma}_\ell}\biggr) 
\p^{k+1}\Bigl(\frac{\p p_\alpha}{\p v_\rho}\Bigr) 
G^{\alpha\beta} \p^{\ell+1} \Bigl(\frac{\p p_\beta}{\p v_\sigma}\Bigr) \nn\\
& \quad + \frac{\e^2}2 \sum_{\alpha,\beta,\rho=1}^n \sum_{k\ge0} \frac{\p \Delta F^M}{\p v^{\rho}_k} 
\p^{k+1} \Bigl(\Bigl(\nabla \Bigl(\frac{\p p_\alpha}{\p \lambda}\Bigr) \cdot 
\nabla \Bigl(\frac{\p p_\beta}{\p \lambda}\Bigr) \cdot \p({\bf v})\Bigr)^\rho \, \Bigr) G^{\alpha\beta} , \label{loopM1208}
\end{align}
together with
\begin{align}
& \sum_{\alpha=1}^n \sum_{k\ge1} k v^{\alpha}_k \frac{\p F_g^M}{\p v^{\alpha}_k} = (2g-2) F_g^M + \delta_{g,1} \frac{n}{24}, \quad g\geq1.
\label{dilatonjet1208}
\end{align}
Here $F_g^M=F_g^M({\bf v}, {\bf v}_1, \dots, {\bf v}_{3g-2})$, $g\geq1$, and $\Delta F^M = \sum_{g\geq1} \e^{2g-2} F_g^M$. 

\smallskip

\noindent {\bf Theorem B} (Dubrovin--Zhang~\cite{DZ-norm}).
{\it Let $M$ be a semisimple calibrated Frobenius manifold. Equations~\eqref{loopM1208}--\eqref{dilatonjet1208} 
have a unique solution up to an additive pure constant. Here, $\p=\sum_{\alpha=1}^n\sum_{k\ge0} v^{\alpha}_{k+1} \frac{\p}{\p v^\alpha_k}$. Moreover,  the unique solution $F_g^M$ ($g\ge2$) depends  
polynomially on $v^{\alpha}_2$, \dots, $v^{\alpha}_{3g-2}$ and rationally on $v^{\alpha}_1$, $\alpha=1,\dots,n$.}

\smallskip

The topological partition function $Z^{\rm top}(\bt;\e)$ 
has the following properties:
\begin{itemize}
\item[i)] It satisfies the Virasoro constraints~\eqref{virazM1213};
\item[ii)] It satisfies the dilaton equation~\eqref{dilatonzM1213};
\item[iii)] It is a  tau-function for the DZ hierarchy of~$M$ (called 
the topological tau-function);
\item[iv)] Its logarithm has the genus expansion;
\item[v)] The genus zero part has the expression~\eqref{fm01213};
\item[vi)] The genus $g$ part $(g\ge1)$ has the $(3g-2)$ jet-variable representation.
\end{itemize}
We end the appendix in mentioning two facts: 1. the above properties i), ii), iv), v), vi) uniquely determine $Z^{\rm top}({\bf t};\e)$. 
2. the above properties ii), iii) together with the $k=-1$ case of the property i) uniquely determine $Z^{\rm top}({\bf t};\e)$.
\end{appendix}

\end{document}